\documentclass[conference, 10pt, twocolumn]{ieeeconf}

\IEEEoverridecommandlockouts
\usepackage[usenames]{color}
\usepackage{enumerate}
\usepackage{url}
\usepackage{subfigure}
\usepackage{amsfonts,mathrsfs}
\usepackage{amssymb,amsmath}
\usepackage{verbatim}
\usepackage{acronym}
\usepackage{mathtools}
\usepackage{cite}
\usepackage{graphicx}
\usepackage{algorithm}
\usepackage[noend]{algpseudocode}

\def\fskip#1{}

\newtheorem{theorem}{Theorem}

\newtheorem{corollary}{Corollary}

\newtheorem{definition}{Definition}
\newtheorem{example}{Example}

\newtheorem{lemma}{Lemma}

\newtheorem{proposition}[theorem]{Proposition}
\newtheorem{remark}{Remark}

\renewcommand{\theenumi}{\arabic{enumi}}

\def\1{{\bf 1}}

\newcommand{\remove}[1]{}

\begin{document}
\title{Maximizing Social Welfare Subject to Network Externalities: A Unifying Submodular Optimization Approach}
\author{\authorblockN{S. Rasoul Etesami}
 \authorblockA{Department of Industrial and Systems Engineering \& Coordinated Science Lab\\ University of Illinois Urbana-Champaign,  Urbana, IL 61801. (etesami1@illinois.edu)}
\thanks{This material is based upon work supported by the Air Force Office of Scientific Research under award number FA9550-23-1-0107 and the NSF CAREER Award under Grant No. EPCN-1944403.}
}
\maketitle
\begin{abstract}
We consider the problem of allocating multiple indivisible items to a
set of networked agents to maximize the social welfare subject to network externalities. Here, the social welfare is given by the sum of agents' utilities and externalities capture the effect that one user of an item has on the item's value to others.  We first provide a general formulation that captures some of the existing models
as a special case. We then show that the maximum social welfare
problem benefits some nice diminishing or increasing marginal return properties.
That allows us to devise polynomial-time approximation algorithms
using the Lov\'asz extension and multilinear extension of the objective functions.
Our principled approach recovers or improves some of the
existing algorithms and provides a simple and unifying framework
for maximizing social welfare subject to network externalities.
\end{abstract}

\section{Introduction} 
externality (also called network effect) is the effect that one user of a good or service has on the product's value to other people. Externalities exist in many network systems such as social, economic, and cyber-physical networks and can substantially affect resource allocation strategies and outcomes. In fact, due to the rapid proliferation of online social networks such as Facebook, Twitter, and LinkedIn, the magnitude of such network effects has been increased to an entirely new level \cite{cao2015pricing}. Here are just a few examples.

\emph{Allocation of Networked Goods}: Many goods have higher values when used in conjunction with others \cite{haghpanah2013optimal}. For instance, people often derive higher utility when using the same product such as cellphones (Figure \ref{Fig:motivation}). One reason is that companies often provide extra benefits for those who adopt their products. Another reason is that the users who buy the same product can share many benefits, such as installing similar Apps or sending free messages. Such products are often referred to as networked goods and are said to exhibit \emph{positive} network externalities. Since each individual wants to hold one product and has different preferences about different products, a natural objective from a managerial perspective is to assign one product to each individual to maximize social welfare subject to network externalities. In other words, we want to maximize social welfare by partitioning the individuals into different groups where the members of each group are assigned the same product.    

\emph{Cyber-Physical Network Security}: An essential task in cyber-physical security is that of providing a resource
allocation mechanism for securing the operation of a set of networked agents (e.g., servers, computers, or data centers) despite external malicious attacks \cite{grossklags2008secure}. One way of doing that is to allocate a security resource to each agent (e.g., by installing one type of antivirus software on each server). Moreover, by extending the security resources to include the ``non-secure" resource, we may assume that an agent who is not protected uses the non-secure recourse. Since the agents are interconnected, the compromise of one agent puts its neighbors at higher risk, and such a failure can cascade over the entire network. As a result, deciding what security resource (including the non-secure resource) is assigned to an agent will indirectly affect all the others. Therefore, an efficient allocation of security resources among the agents who experience network externalities is a major challenge in network security. 

\emph{Distributed Congestion Networks}: There are many instances of networked systems, such as transportation \cite{roughgarden2005selfish} or data-placement networks \cite{etesami2016pure,baev2008approximation}, in which the cost of agents increases as more agents use the same resource. For instance, in data-placement networks such as web-caches or peer-to-peer networks, an important goal is to store at each node (agent) of a network one copy of some file (resource) to minimize the sum of agents' delay costs to access all files \cite{etesami2020complexity}. As more agents store the same file, the data distribution across the network will be less balanced, hence increasing the delay cost for obtaining some files \cite{etesami2020complexity,baev2008approximation}. Similarly, in transportation networks, as more drivers (agents) use the same path (resource), the traffic congestion on that path will increase, hence increasing the travel time and energy consumption for the drivers (Figure \ref{Fig:motivation}). Therefore, a natural goal is to assign each driver to one path to minimize the overall congestion cost in the network \cite{roughgarden2005selfish,etesami2020smart}. Such network effects are often referred to as \emph{negative} externalities and have been studied under the general framework of congestion games in both centralized or game-theoretic settings \cite{etesami2020complexity,etesami2020smart,blumrosen2007welfare,rosenthal1973class,milchtaich1996congestion,roughgarden2005selfish}.   

\begin{figure*}
\vspace{0.3cm}
  \begin{center}
    \includegraphics[totalheight=.15\textheight,
width=.25\textwidth,viewport= 570 0 920 350]{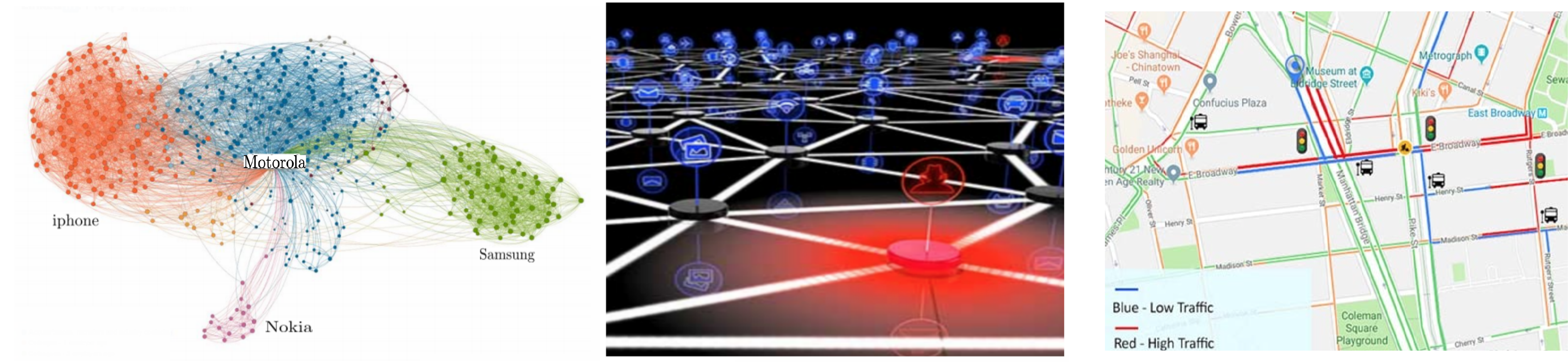}
  \end{center}
  \caption{\footnotesize{The left figure shows an instance of network goods with four different cellphone products. Individuals tend to buy a product that is adopted by most of their friends. The middle figure shows networked servers that are highly interconnected and an adversary who has compromised one of them and hence influences all others. The right figure illustrates the GPS map of a traffic network. As more drivers use the same road, they will negatively influence each others' travel time.}}\label{Fig:motivation}
\end{figure*} 

Motivated by the above, and many other similar examples, our objective in this paper is to study allocation problems when agents exhibit network externalities. While this problem has been significant in the past literature \cite{haghpanah2013optimal,candogan2012optimal,akhlaghpour2010optimal,bhattacharya2011allocations}, these results are mainly focused on allocating and pricing of copies of a \emph{single} item; other than a handful of results \cite{etesami2020complexity,de2012finding,bhalgat2012mechanisms}, the problem of maximizing the social welfare by allocating multiple items subject to network externalities has not been well-studied before. Therefore, we consider the more realistic situation with multiple competing items and when the agents in the network are demand-constrained. Moreover, we consider both positive and negative externalities with linear, convex, and concave functions. Such a comprehensive study allows us to capture more complex situations, such as traffic routing, where the change in an agent's cost depends nonlinearly (e.g., using a polynomial function \cite{roughgarden2005selfish}) on the number of agents that use the same route. 

\subsection{Related Work}

There are many papers that consider resource allocation under various network externality models. For example, negative externalities have been studied in routing \cite{roughgarden2005selfish,etesami2020smart}, facility location \cite{etesami2020complexity,etesami2017price}, welfare maximization in congestion games \cite{blumrosen2007welfare}, and monopoly pricing over a social network \cite{bhattacharya2011allocations}. On the other hand, positive externalities have been addressed in the context of mechanism design and optimal auctions \cite{bhalgat2012mechanisms,haghpanah2013optimal}, congestion games with positive externalities  \cite{de2012finding,blumrosen2007welfare}, and pricing networked goods \cite{feldman2013pricing,candogan2012optimal}. There are also some results that consider \emph{unrestricted} externalities where a mixture of both positive and negative externalities may exist in the network \cite{blumrosen2007welfare,chakrabarty2005fairness}. However, those results are often for simplified \emph{anonymous} models in which the agents do not care about the identity of others who share the same resource with them. One reason is that for unrestricted and non-anonymous externalities, maximizing the social welfare with $n$ agents is $n^{1-\epsilon}$-inapproximable for any $\epsilon>0$ \cite{de2012finding,blumrosen2007welfare}. Therefore, in this work, we consider maximizing social welfare with \emph{non-anonymous} agents but with either positive or negative externalities.  

Optimal resource allocation subject to network effects is typically NP-hard and even hard to approximate \cite{dekeijzer2014externalities,haghpanah2013optimal,etesami2020complexity,blumrosen2007welfare}. Therefore, a large body of past literature has been devoted to devising polynomial-time approximation algorithms with a good performance guarantee. Maximizing social welfare subject to network externalities can often be cast as a special case of a more general combinatorial welfare maximization problem \cite{lehmann2006combinatorial}. However, combinatorial welfare maximization with general valuation functions is hard to approximate to within a factor better than $\sqrt{n}$, where $n$ is the number of items  \cite{blumrosen2007combinatorial}. Therefore, to obtain improved approximation algorithms for the special case of welfare maximization with network externalities, one must rely on more tailored algorithms that take into account the special structure of the agents' utility functions. 

Another closely related problem on resource allocation under network effects is submodular optimization \cite{chekuri2014submodular,calinescu2007maximizing}. The reason is that utility functions of the networked agents often exhibit diminishing return property as more agents adopt the same product. That property makes a variety of submodular optimization techniques quite amenable to design improved approximation algorithms. While this connection has been studied in the past literature for the special case of a single item \cite{haghpanah2013optimal}, it has not been leveraged for the more complex case of multiple items. Unlike earlier literature \cite{de2012finding,bhalgat2012mechanisms,dekeijzer2014externalities}, our first contribution is to show that multi-item welfare maximization under network externalities can be formulated a special case of minimum submodular cost allocation (MSCA) problem \cite{chekuri2011submodular,chekuri2011approximation}. In MSCA, we are given a finite ground set $V$ and $k$ nonnegative submodular set functions $f_i, i=1,\ldots,k$, and the goal is to partition $V$ into $k$ (possibly empty) sets $S_1,\ldots,S_k$ such that the sum $\sum_{i=1}^kf_i(S_i)$ is minimized. The authors in \cite[Theorem 2]{chekuri2011submodular} used Lova\'sz extension and the Kleinberg-Tardos (KT) rounding scheme of \cite{kleinberg2002approximation} to develop an $O(\log(|V|))$-approximation algorithm for MSCA with monotone submodular cost functions. In general, MSCA is inapproximable within any multiplicative factor even in very restricted settings \cite{ene2014hardness}, and for monotone submodular cost functions, the poly-logarithmic approximation factor is the best one can hope for (as nearly matching logarithmic lower bounds are known \cite{ene2014hardness}). Therefore, instead of adapting this general framework naively to our problem setting, which can only deliver a poly-logarithmic approximation factor, as our second contribution, we exploit the special structure of the multi-item welfare maximization to obtain constant factor approximation algorithms using refined analysis of the KT randomized rounding. We should mention that there could be alternative reductions between multi-item welfare maximization and special cases of MSCA, such as submodular multi-way partition \cite{chekuri2011approximation}. However, we believe that our concise reduction is very natural and requires solving a small-size concave program, which can potentially be applied to more general problems in the above category.

A further generalization of MSCA has been studied in the past literature under the framework of multi-agent submodular optimization (MASO) \cite{santiago2018multi}, in which given submodular cost functions $f_i, i=1,\ldots,k$, the goal is to solve $\min \sum_{i=1}^kf_i(S_i)$, subject to the constraint that the disjoint union of $S_i, i=1,\ldots,k$, must belong to a given family $\mathcal{F}$ of feasible sets. When $\mathcal{F}=\{V\}$, where $V$ is the ground set, MASO reduces to the MSCA, and thus all the inapproximability results for MSCA also apply to MASO. Finally, an extension of MASO has been studied under multivariate submodular optimization (MVSO) \cite{santiago2019multivariate}, in which the objective function has a more general form of $f(S_1,\ldots,S_k)$, where $f$ captures some notion of submodularity across its arguments. Instead of using these general frameworks naively as a black-box, we will leverage the special structure of the agents' utility functions and new ideas from submodular optimization to devise improved approximation algorithms for maximizing the social welfare subject to network externalities.

\subsection{Contributions and Organization}

We first provide a general model for the maximum social welfare problem with multiple items subject to network externalities and show that the proposed model subsumes some of the existing ones as a special case. We then show that the proposed model can be formulated as a special case of multi-agent submodular optimization. Leveraging this connection and the special structure of agents' utility functions, we devise unified approximation algorithms for the multi-item maximum social welfare problem using continuous extensions of the objective functions and refined analysis of various rounding techniques such as KT randomized rounding and fair contention resolution scheme. While some of such rounding algorithms were developed for applications such as metric labeling, our work is the first to show that variants of these techniques can be used effectively to analyze the multi-item maximum social welfare problem subject to network externalities. Our principled approach not only recovers or improves the state-of-the-art approximation guarantees but also can be used for devising approximation algorithms with potentially more complex constraints.

The paper is organized as follows. In Section \ref{sec:formulation}, we formally introduce the multi-item maximum social welfare problem subject to network externalities. In Section \ref{sec:prelim}, we provide some preliminary results from submodular optimization for later use. In Section \ref{sec:nonpositive}, we consider the problem of maximum social welfare under negative concave externalities and provide a constant-factor approximation algorithm for that problem. In Section \ref{sec:convex}, we consider positive polynomial externalities as well as more general positive convex externality functions and devise improved approximation algorithms in terms of the degree of the polynomials and the curvature of the externality functions. Finally, we extend our results to devise approximation algorithms for positive concave externality functions in Section \ref{sec:concave}. We conclude the paper by identifying some future research directions in Section \ref{sec:conclusion}.  

\subsection{Notations}
We adopt the following notations throughout the paper: For a positive integer $n\in\mathbb{Z}^+$ we set $[n]:=\{1,2,\ldots,n\}$. We use bold symbols for vectors and matrices. For a matrix $\boldsymbol{x}=(x_{ji})$ we use $\boldsymbol{x}_j$ to refer to its $j$th row and $\boldsymbol{x}_i$ to refer to its $i$th column. Given a vector $\boldsymbol{v}$ we denote its transpose by $\boldsymbol{v}'$. We let $\boldsymbol{1}$ and $\boldsymbol{0}$ be column vectors of all ones and all zeros, respectively.
 
\section{Problem Formulation}\label{sec:formulation} 

Consider a set $[n]=\{1,\ldots,n\}$ of agents and a set $[m]=\{1,\ldots,m\}$ of distinct indivisible items (resources). There are unlimited copies of each item $i\in [m]$; however, each agent can receive at most one item. For any ordered pair of agents $(j,k)$ and any item $i$, there is a weight $a^i_{jk}\in \mathbb{R}$ indicating the amount by which the utility of agent $j$ gets influenced from agent $k$, given that both agents $j$ and $k$ receive the same item $i$. In particular, for $j=k$, the parameter $a^i_{jj}\ge 0$ captures intrinsic valuation of item $i$ by agent $j$. If $a^i_{jk}\ge 0, \forall i,j,k$ with $j\neq k$, we say that the agents experience \emph{positive} externalities. Otherwise, if $a^i_{jk}\le 0, \forall i,j,k$ with $j\neq k$, the agents experience \emph{negative} externalities. We refer to Figure \ref{Fig:min:realization} for an illustration of network externality weights.\footnote{The fact that intrinsic valuations $a^i_{jj}$ are nonnegative implies that agents always derive positive utilities by receiving an item. However, depending on positive or negative externalities, an agent's utility increases or decreases as more agents receive the same item.}

\begin{figure}
\vspace{0.3cm}
  \begin{center}
    \includegraphics[totalheight=.16\textheight,
width=.25\textwidth,viewport=100 0 700 600]{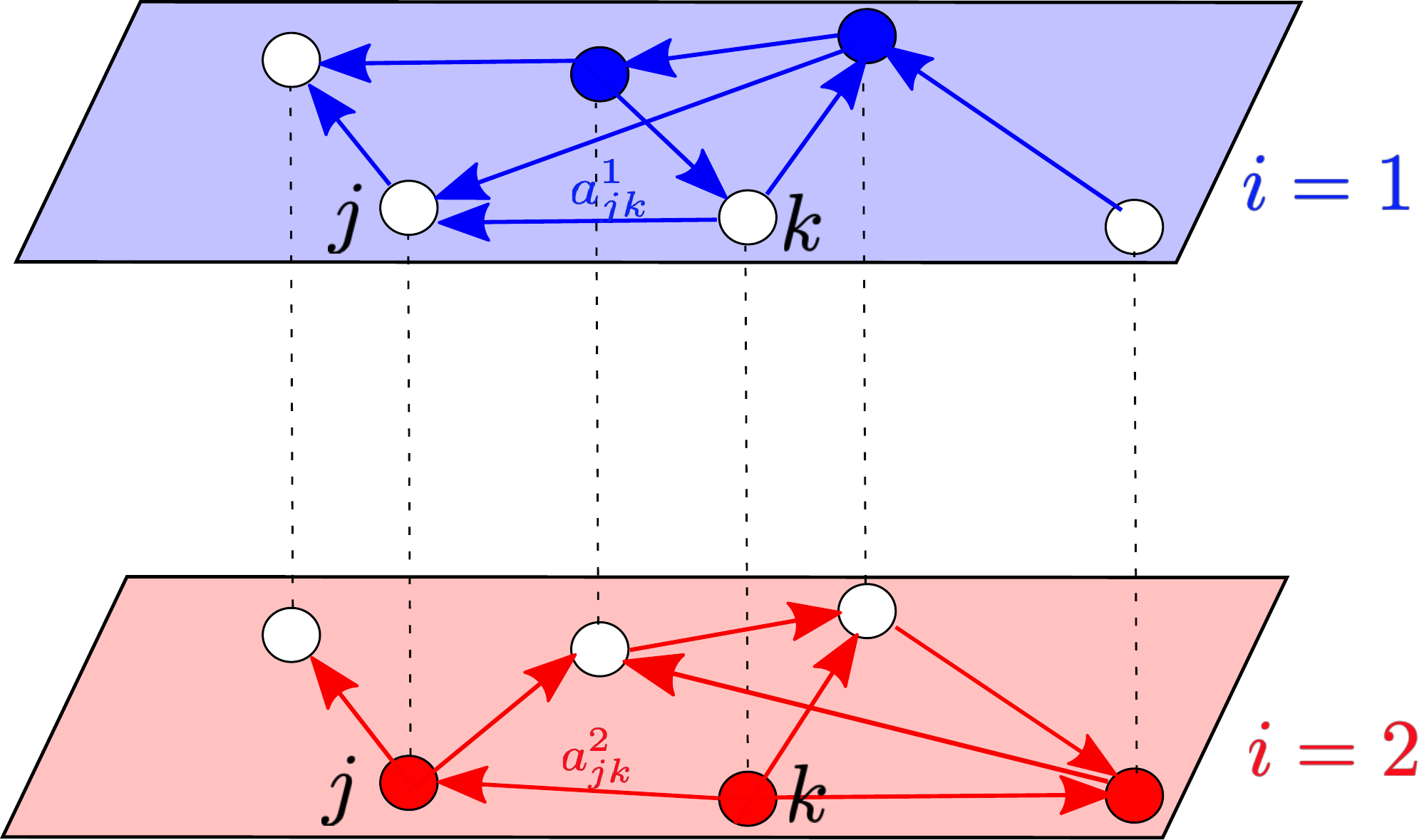}
  \end{center}
  \caption{\footnotesize{An instance of the MSW with $n=6$ agents and $m=2$ items: blue item $(i=1)$ and red item $(i=2)$. Each layer represents the directed influence graph between the agents for that specific item. The influence weights are captured by $a^1_{jk}, a^2_{jk}, \forall j,k$. If there is no edge between two agents $j$ and $k$ in an item layer $i$, it means that $a^i_{jk}=0$. Note that each agent can adopt at most one item. In the above figure, each of agents $j$ and $k$ is allocated a red item.}}\label{Fig:min:realization}
\end{figure}

Let $S_i$ denote the set of agents that receive item $i$ in a given allocation. For any $j\in S_i$, the utility that agent $j$ derives from such an allocation is given by $f_{ij}\big(\sum_{k\in S_i}a^i_{jk}\big)$, where $f_{ij}:\mathbb{R}\to \mathbb{R}$ with $f_{ij}(0)=0$ is a nondecreasing function. Depending on whether the functions $f_{ij}$ are linear, convex, or concave, we will refer to them as \emph{linear externalities}, \emph{convex externalities}, or \emph{concave externalities}. In the maximum social welfare (MSW) problem, the goal is to assign at most one item to each agent in order to maximize the social welfare. In other words, we want to find disjoint subsets $S_1,\ldots,S_m$ of agents such that $\cup_{i=1}^m S_i\subseteq [n],\ S_i\cap S_{i'}=\emptyset\ \forall i\neq i'$ to maximize the sum of agents' utilities given by
\begin{align}\label{eq:general-welfare-convex}
\max_{S_1,\ldots,S_m} \sum_{i=1}^{m}\sum_{j\in S_i} f_{ij}\big(\sum_{k\in S_i}a^i_{jk}\big).
\end{align}
We note that $[n]\setminus\cup_{i=1}^m S_i$ is the set of agents that do not receive any item. Such agents are assumed to derive zero utility and hence contribute zero to the objective function \eqref{eq:general-welfare-convex}. Now let us define binary variables $x_{ji}\in \{0,1\}$, where $x_{ji}=1$ if and only if item $i$ is assigned to agent $j$. Using the fact that $f_{ij}(0)=0\ \forall i,j$, the MSW \eqref{eq:general-welfare-convex} can be formulated as the following integer program (IP):
\begin{align}\label{eq:IP-convex}
&\max \sum_{i,j}f_{ij}\big(\sum_{k=1}^na^i_{jk}x_{ji}x_{ki}\big)\cr 
&\qquad \sum_{i=1}^m x_{ji}\leq 1\ \forall j\in [n],\cr
&\qquad  \ x_{ji}\in\{0,1\}\ \forall i\in [m], j\in [n].
\end{align}
In particular, the IP \eqref{eq:IP-convex} can be written in a compact form as
\begin{align}\label{eq:SM}   
 \max \Big\{\sum_{i=1}^{m}f_i(\boldsymbol{x}_i): \ \sum_{i=1}^m\boldsymbol{x}_i\leq \boldsymbol{1}, \boldsymbol{x}_i\in \{0,1\}^n, \forall i\Big\},
\end{align}
where for any $i\in [m]$, we define $\boldsymbol{x}_i$ to be the binary column vector $\boldsymbol{x}_i=(x_{1i},\ldots,x_{ni})'$, and $f_i:\{0,1\}^n\to \mathbb{R}$ is given by 
\begin{align}\nonumber
f_i(\boldsymbol{x}_i)=\sum_{j}f_{ij}\big(\sum_{k=1}^na^i_{jk}x_{ji}x_{ki}\big).
\end{align}
We note that the objective function in IP \eqref{eq:SM} is separable across variables $\boldsymbol{x}_i, i\in [m]$.

\begin{example}\label{ex:linear}
For the special case of linear functions $f_{ij}(y)=y\ \forall i,j$, the objective function in \eqref{eq:general-welfare-convex} becomes $\sum_{i}\sum_{(j,k)\in S_i}a^i_{jk}$, hence recovering the optimization problem studied in \cite{de2012finding}. We refer to such externality functions as linear externalities. 
\end{example}

\begin{example}\label{ex:convex}
Let $\mathcal{G}=([n], \mathcal{E})$ be a fixed directed graph among the agents, and denote the set of in neighbors of agent $j$ by $N_j$. In a special case when each agent treats all of its in neighbors equally regardless of what item they use (i.e.,  for every item $i$ we have $a^i_{jk}=1$ if $k\in N_j$ and $a^i_{jk}=0$ otherwise), the objective function in \eqref{eq:general-welfare-convex} becomes $\sum_{i}\sum_{j\in S_i} f_{ij}\big(|S_i\cap N_j|\big)$. Therefore, we recover the maximum social welfare problem studied in \cite{bhalgat2012mechanisms}. For this special case, it was shown in \cite[Theorem 3.9]{bhalgat2012mechanisms} that when the externality functions are convex and bounded above by a polynomial of degree $d$, one can find an $2^{O(d)}$-approximation for the optimum social welfare allocation. In this work, we will improve this result for the more general setting of \eqref{eq:general-welfare-convex}.  
\end{example}

\section{Preliminary Results}\label{sec:prelim}

This section provides some definitions and preliminary results, which will be used later to establish our main results. We start with the following definition.\begin{definition}
Given a finite ground set $N$, a set function $f:2^N\to \mathbb{R}$ is called \emph{submodular} if and only if $f(A)+f(B)\ge f(A\cup B)+f(A\cap B)$ for any $A, B\subseteq N$. Equivalently, $f$ is submodular if for any two nested subsets $A\subseteq B$ and any $i\notin B$, we have $f(A\cup\{i\})-f(A)\ge f(B\cup\{i\})-f(B)$. A set function $f:2^N\to \mathbb{R}$ is called \emph{supermodular} if $-f$ is submodular. A set function $f:2^N\to \mathbb{R}$ is called \emph{monotone} if $f(A)\leq f(B)$ for $A\subseteq B$.   
\end{definition} 

\subsection{Lov\'asz Extension}
Let $N$ be a ground set of cardinality $n$. Each real-valued set function on $N$ corresponds to a function $f:\{0,1\}^n\to \mathbb{R}$ over the vertices of hypercube $\{0,1\}^n$, where each subset is represented by its binary characteristic vector. Therefore, by abuse of notation, we use $f(S)$ and $f(\chi_S)$ interchangeably where $\chi_S\in \{0, 1\}^n$ is the characteristic vector of the set $S\subseteq N$. Given a set function $f:\{0,1\}^n\to \mathbb{R}$, the \emph{Lov\'asz extension} of $f$ to the continuous unit cube $[0,1]^n$, denoted by $f^L:[0,1]^n\to \mathbb{R}$, is defined by 
\begin{align}\label{eq:Lovasz-ext}
f^L(\boldsymbol{x}):=\mathbb{E}_{\theta}[f(\boldsymbol{x}^{\theta})]=\int_{0}^{1}f(\boldsymbol{x}^{\theta})d\theta,
\end{align}
where $\theta\in [0, 1]$ is a uniform random variable, and $\boldsymbol{x}^{\theta}$ for a given vector $\boldsymbol{x}\in [0, 1]^n$ is defined as: $x_i^{\theta}=1$ if $x_i\ge \theta$, and $x_i^{\theta}=0$, otherwise. In other words, $\boldsymbol{x}^{\theta}$ is a random binary vector obtained by rounding all the coordinates of $\boldsymbol{x}$ that are above $\theta$ to $1$, and the remaining ones to $0$. In particular, $f^L(\boldsymbol{x})$ is equal to the expected value of $f$ at the rounded solution $\boldsymbol{x}^{\theta}$, where the expectation is with respect to the randomness posed by $\theta$. It is known that the Lov\'asz extension $f^L$ is a convex function of $\boldsymbol{x}$ if and only if the corresponding set function $f$ is submodular \cite{Lovasz1983submodular}. This property makes the Lov\'asz extension a suitable continuous extension for submodular minimization.

\subsection{Multilinear Extension}
As mentioned earlier, the Lov\'asz extension provides a convex continuous extension of a submodular function, which is not very useful for maximizing a submodular function. For the maximization problem, one can instead consider another continuous extension known as \emph{multilinear extension}. The multilinear extension of a set function $f:2^N\to \mathbb{R}$ at a given vector $\boldsymbol{x}\in [0, 1]^n$, denoted by $f^M(\boldsymbol{x})$, is given by the expected value of $f$ at a random set $R(\boldsymbol{x})$ that is sampled from the ground set $N$ by including each element $i$ to $R(\boldsymbol{x})$ independently with probability $x_i$, i.e., 
\begin{align}\nonumber
f^M(\boldsymbol{x})=\mathbb{E}\big[f\big(R(\boldsymbol{x})\big)\big]=\sum_{R\subseteq N}f(R)\prod_{i\in R}x_i\prod_{i\notin R}(1-x_i).
\end{align}
One can show that the Lov\'asz extension is always a lower bound for the multilinear extension, i.e., $f^L(\boldsymbol{x})\leq f^M(\boldsymbol{x}), \forall \boldsymbol{x}\in [0, 1]^n$. Moreover, at any binary vector $\boldsymbol{x}\in \{0,1\}^n$, we have $f^L(\boldsymbol{x})=f(\boldsymbol{x})=f^M(\boldsymbol{x})$. In general, the multilinear extension of a submodular function is neither convex nor concave. However, it is known that there is a polynomial-time \emph{continuous greedy} algorithm that can approximately maximize the multilinear extension of a nonnegative submodular function subject to a certain class of constraints. That result is stated in the following lemma.

\begin{lemma}\label{lemm:submodular-multilinear}\cite[Theorems I.1 \& I.2]{feldman2011unified}
For any nonnegative submodular function $f:2^N\to \mathbb{R}_+$, down-monotone solvable polytope\footnote{A polytope $\mathcal{P}\subseteq [0, 1]^n$ is solvable if linear functions can be maximized over it in polynomial time. It is down-monotone if $\boldsymbol{x}\in \mathcal{P}$ and $\boldsymbol{0}\leq \boldsymbol{y}\leq \boldsymbol{x}$ (coordinate-wise) implies that $\boldsymbol{y}\in \mathcal{P}$.} $\mathcal{P}\subseteq [0, 1]^n$, there is a polynomial-time continuous greedy algorithm that finds a point $\boldsymbol{x}^*\in \mathcal{P}$ such that $f^M(\boldsymbol{x}^*)\ge \frac{1}{e} f(OPT)$, where $OPT$ is the optimal integral solution to the maximization problem $\max_{\boldsymbol{x}\in \mathcal{P}\cap \{0, 1\}^n}f^M(\boldsymbol{x})$. If in addition, the submodular function $f$ is monotone, the approximation guarantee can be improved to $f^M(\boldsymbol{x}^*)\ge (1-\frac{1}{e}) f(OPT)$.  
\end{lemma} 

According to Lemma \ref{lemm:submodular-multilinear}, the multilinear extension provides a suitable relaxation for devising an approximation algorithm for submodular maximization. The reason is that one can first approximately solve the multilinear extension in polynomial time and then round the solution to obtain an approximate integral feasible solution. 

\subsection{Fair Contention Resolution}
Here, we provide some background on a general randomized rounding scheme known as \emph{fair contention resolution} that allows one to round a fractional solution to an integral one while preserving specific properties. Intuitively, given a fractional solution to a resource allocation problem, one ideally wants to round the solution to an integral allocation such that each item is allocated to only one agent. However, a natural randomized rounding often does not achieve that property as multiple agents may receive the same item. To resolve that issue, one can use a ``contention resolution scheme," which determines which agent should receive the item while losing at most a constant factor in the objective value. 

More precisely, suppose $n$ agents compete for an item independently with probabilities $p_1, p_2, \ldots, p_n$. Denote by $A$ the random set of agents who request the item in the first phase, i.e., $\mathbb{P}(i\in A)=p_i$ independently for each $i$. In the second phase, If $|A|\leq 1$, we do not make any change to the allocation. Otherwise, allocate the item to each agent $i\in A$ who requested the item in the first phase with probability
\begin{align}\nonumber
r_{iA}=\frac{1}{\sum_{j=1}^n p_j}\Big(\sum_{k\in A\setminus \{i\}}\frac{p_k}{|A|-1}+\sum_{k\notin A}\frac{p_k}{|A|}\Big).
\end{align}
Note that for any $A\neq \emptyset$, we have $\sum_{i\in A}r_{iA}=1$, so that after the second phase, the item is allocated to exactly one agent with probability $1$. The importance of such a fair contention resolution scheme is that if the item was requested in the first phase by an agent, then after the second phase, that agent still receives the item with probability at least $1-\frac{1}{e}$. More precisely, it can be shown that \cite{feige2010submodular}:  
\begin{lemma}\cite[Lemma 1.5]{feige2010submodular}\label{lemm:contention}
Conditioned on agent $k$ requesting the item in the first phase, she obtains it after the second phase with probability exactly 
\begin{align}\nonumber
\frac{1-\prod_{j=1}^n (1-p_j)}{\sum_{j=1}^np_j}\ge 1-\frac{1}{e}.
\end{align}
\end{lemma}


\section{MSW with Negative Concave Externalities}\label{sec:nonpositive}

In this section, we consider maximizing the social welfare with negative concave externalities and provide a constant factor approximation algorithm by reducing that problem to submodular maximization subject to a matroid constraint.

\begin{lemma}\label{lemm:convex-sub}
Given nondecreasing concave externality functions $f_{ij}:\mathbb{R}\to \mathbb{R}$, intrinsic valuations $a^i_{jj}\ge 0, \forall i,j$, and negative externality weights $a^i_{jk}\leq 0, \forall i,j\neq k$, the objective function in \eqref{eq:general-welfare-convex} is a submodular set function.   
\end{lemma}
\begin{proof}
Let $f_i(S_i)=\sum_{j\in S_i} f_{ij}\big(\sum_{k\in S_i}a^i_{jk}\big)$, and note that the objective function in \eqref{eq:general-welfare-convex} can be written in a separable form as $\sum_{i=1}^{m}f_i(S_i)$. Thus, it is enough to show that each $f_i$ is a submodular set function. For any  $A\subseteq B, \ell\notin B$, we can write
\begin{align}\label{eq:sub-mod}
&f_i(B\cup\{\ell\})-f_i(B)\cr 
&=\sum_{j\in B\cup\{\ell\}} f_{ij}\big(\!\!\!\sum_{k\in B\cup\{\ell\}}\!\!\!a^i_{jk}\big)-\sum_{j\in B} f_{ij}\big(\sum_{k\in B}a^i_{jk}\big)\cr 
&=\sum_{j\in B} \Big(f_{ij}\big(\!\!\!\sum_{k\in B\cup\{\ell\}}\!\!\!a^i_{jk}\big)-f_{ij}\big(\sum_{k\in B}a^i_{jk}\big)\Big)+f_{i\ell}\big(\!\!\!\sum_{k\in B\cup\{\ell\}}\!\!\!a^i_{\ell k}\big)\cr 
&\leq \sum_{j\in A} \Big(f_{ij}\big(\!\!\!\sum_{k\in B\cup\{\ell\}}\!\!\!a^i_{jk}\big)-f_{ij}\big(\sum_{k\in B}a^i_{jk}\big)\Big)+f_{i\ell}\big(\!\!\!\sum_{k\in A\cup\{\ell\}}\!\!\!a^i_{\ell k}\big)\cr
&\leq \sum_{j\in A} \Big(f_{ij}\big(\!\!\!\sum_{k\in A\cup\{\ell\}}\!\!\!a^i_{jk}\big)-f_{ij}\big(\sum_{k\in A}a^i_{jk}\big)\Big)+f_{i\ell}\big(\!\!\!\sum_{k\in A\cup\{\ell\}}\!\!\!a^i_{\ell k}\big)\cr
&=f_i(A\cup\{\ell\})-f_i(A).
\end{align} 
The first inequality holds by the monotonicity of functions $f_{ij}$ and by $A\subseteq B, \ell\notin B$ (note that since $a^i_{jk}\leq 0, j\neq k$, each of the summands in the first summation is nonpositive). The second inequality in \eqref{eq:sub-mod} follows from concavity of the functions $f_{ij}$. More precisely, given any $j\in A$, let $\sum_{k\in B\setminus A}a^i_{jk}=d$, $\sum_{k\in A\cup\{\ell\}}a^i_{jk}=p$, and $\sum_{k\in A}a^i_{jk}=q$, where we note that $d\leq 0$ and $p\leq q$. By concavity of $f_{ij}$ we have $f_{ij}(q)-f_{ij}(q+d)\leq f_{ij}(p)-f_{ij}(p+d)$, or equivalently $f_{ij}(p+d)-f_{ij}(q+d)\leq f_{ij}(p)-f_{ij}(q)$, which is exactly the second inequality in \eqref{eq:sub-mod}.  
\end{proof}

Let us now consider the MSW with negative concave externalities. However, to assure that the maximization problem from the lens of approximation algorithm is well-defined, we assume that for any feasible assignment of items to the agents, the objective value in \eqref{eq:general-welfare-convex} is nonnegative. Otherwise, the maximization problem may have a negative optimal value, hence hindering the existence of an approximation algorithm. In fact, if a feasible allocation $(S_1,\ldots,S_m)$ returns a negative objective value, then by unassigning all the items, one can obtain the trivial higher objective value of $0$. Therefore, without loss of generality, we may restrict our attention to allocation profiles for which the objective value \eqref{eq:general-welfare-convex} is nonnegative.

\begin{theorem}\label{thm:multi-negative}
There is a randomized $e$-approximation algorithm for the MSW \eqref{eq:general-welfare-convex} with negative concave externalities.
\end{theorem}
\begin{proof}
Let us consider the IP formulation \eqref{eq:SM} for the MSW and note that by Lemma \ref{lemm:convex-sub}, the objective function $f(\boldsymbol{x})=\sum_i f_i(\boldsymbol{x}_i)$ is a nonnegative and submodular function. Here, $\boldsymbol{x}$ can be viewed as an $n\times m$ matrix whose $i$th column is given by $\boldsymbol{x}_i$. Using separability of $f(\boldsymbol{x})$, the multilinear relaxation for IP \eqref{eq:SM} is given by
\begin{align}\label{eq:approximate-multilinear}
\max\!\Big\{\!f^M(\boldsymbol{x})\!=\!\!\sum_{i=1}^m f_i^{M}(\boldsymbol{x}_i): \sum_{i=1}^m \boldsymbol{x}_i\leq \boldsymbol{1}, \boldsymbol{x}_i\ge \boldsymbol{0}, \forall i\in [m]\Big\},  
\end{align}  
where we have relaxed the binary constraints $\boldsymbol{x}_i\in\{0,1\}^n$ to $\boldsymbol{x}_i\ge \boldsymbol{0}$. The feasible set $\mathcal{P}=\{\boldsymbol{x}: \sum_i \boldsymbol{x}_i\leq \boldsymbol{1}, \boldsymbol{x}_i\ge 0\}$ is clearly a down-monotone polytope as $\boldsymbol{x}\in \mathcal{P}$ and $\boldsymbol{0}\leq \boldsymbol{y}\leq \boldsymbol{x}$, implies $\boldsymbol{y}\in \mathcal{P}$. Moreover, $\mathcal{P}$ is a solvable polytope as it contains only $m+n$ linear constraints and a total of $mn$ variables. Therefore, using Lemma \ref{lemm:submodular-multilinear}, one can find, in polynomial time, an approximate solution $\boldsymbol{x}^*$ to \eqref{eq:approximate-multilinear} such that $f^M(\boldsymbol{x}^*)\ge \frac{1}{e}f(OPT)$, where $OPT$ denotes the optimal integral solution to IP \eqref{eq:SM}. 

Next, we can round the approximate solution $\boldsymbol{x}^*$ to an integral one $\hat{\boldsymbol{x}}$ by rounding each row of $\boldsymbol{x}^*$ independently using the natural probability distribution induced by that row. More precisely, for each row $j$ (and independently of other rows), we pick entry $(j,i)$ in row $j$ with probability $x^*_{ji}$ and only round that entry to $1$ while setting the remaining entries of row $j$ to $0$. Such a rounding sets at most one entry in each row of the rounded solution to $1$ because $\sum_{i}x^*_{ji}\leq 1$. Since the rounding is done independently across the rows, for any column $i$, the probability that the $j$th entry is set to $1$ is $x^*_{ji}$, which is independent of the other entries in that column. Therefore, $\hat{\boldsymbol{x}}_i$ represents the characteristic vector of a random set $R(\boldsymbol{x}^*_i)\subseteq [n]$, where $j\in R(\boldsymbol{x}^*_i)$ independently with probability $x^*_{ji}$. Moreover, although the rounded solution $\hat{\boldsymbol{x}}$ is correlated across its columns, because the objective function $f(\hat{\boldsymbol{x}})$ is separable across columns, using linearity of expectation and regardless of the rounding scheme we have $\mathbb{E}[f(\hat{\boldsymbol{x}})]=\sum_{i=1}^m\mathbb{E}[f_i(\hat{\boldsymbol{x}}_i)]$. Thus, by definition of the multilinear extension, we have
\begin{align}\nonumber
\mathbb{E}[f(\hat{\boldsymbol{x}})]&=\sum_{i=1}^m \mathbb{E}[f_i(\hat{\boldsymbol{x}}_i)]=\sum_{i=1}^m\mathbb{E}\big[f_i\big(R(\boldsymbol{x}^*_i)\big)\big]\cr 
&=\sum_{i=1}^m f^M_i(\boldsymbol{x}^*_i)=f^M(\boldsymbol{x}^*)\ge \frac{1}{e}f(OPT). 
\end{align}
\end{proof}

\begin{remark}
The constraints in $\mathcal{P}$ define a partition matroid. Subsequently, one can replace the independent rounding scheme in Theorem \ref{thm:multi-negative} by the \emph{pipage} rounding scheme \cite[Lemma B.3]{vondrak2013symmetry} and obtain the same performance guarantee. However, due to the special structure of $\mathcal{P}$, such a complex rounding is not necessary, and one can substantially save in the running time using the proposed independent rounding.
\end{remark}


For the special case of negative weights $a^i_{jk}\leq 0, \forall j\neq k$ and linear externality functions $f_{ij}(y)=y, \forall i,j$, the objective function in IP \eqref{eq:SM} becomes $\sum_{i=1}^m\boldsymbol{x}_i'\boldsymbol{A}_i\boldsymbol{x}_i$, where $\boldsymbol{A}_i=(a^i_{jk}), i\in [m]$ are $n\times n$ weight matrices with nonnegative diagonal entries (due to intrinsic valuations) and negative off-diagonal entries. Applying Theorem \ref{thm:multi-negative} to this special case gives an $e$-approximation algorithm for the MSW with negative linear externalities, which answers a question posed in \cite[Section 7.11]{dekeijzer2014externalities}. In particular, if we further assume that the influence weight matrices $\boldsymbol{A}_i, i\in [m]$ are diagonally dominant, i.e., $\sum_{k=1}^n a^i_{jk}\ge 0, \forall i, j$, then the submodular objective function $\sum_{i=1}^m\boldsymbol{x}_i'\boldsymbol{A}_i\boldsymbol{x}_i$ will also be monotone. In that case, using the second part of Lemma \ref{lemm:submodular-multilinear} one can obtain an improved approximation factor of $1-\frac{1}{e}$.

\section{MSW with Positive Monotone Convex Externalities}\label{sec:convex}

In this section, we consider positive monotone convex externalities and develop polynomial-time approximation algorithms for the maximum social welfare problem. We first state the following lemma that is a counterpart of Lemma \ref{lemm:convex-sub} to the case of positive convex externalities.

\begin{lemma}\label{lemm:convex-supper}
For positive weights $a^i_{jk}\ge 0$ and nondecreasing convex externality functions $f_{ij}:\mathbb{R}_+\to \mathbb{R}_+$, the objective function in \eqref{eq:general-welfare-convex} is a nondecreasing and nonnegative supermodular set function.   
\end{lemma}
\begin{proof}
As in Lemma \ref{lemm:convex-sub}, if we define $f_i(S_i)=\sum_{j\in S_i} f_{ij}\big(\sum_{k\in S_i}a^i_{jk}\big)$, it is enough to show that each $f_i$ is a monotone supermodular set function. The monotonicity and nonnegativity of $f_i$ immediately follows from nonnegativity of weights $a^i_{jk}$, and monotonicity and nonnegativity of $f_{ij}, \forall j\in [n]$. To show supermodularity of $f_{i}$, for any  $A\subseteq B, \ell\notin B$ and similar to Lemma \ref{lemm:convex-sub}, we can write
\begin{align}\label{eq:supper-mod}
&f_i(B\cup\{\ell\})-f_i(B)\cr
&=\sum_{j\in B} \Big(f_{ij}\big(\!\!\!\sum_{k\in B\cup\{\ell\}}\!\!\!a^i_{jk}\big)-f_{ij}\big(\sum_{k\in B}a^i_{jk}\big)\Big)+f_{i\ell}\big(\!\!\!\sum_{k\in B\cup\{\ell\}}\!\!\!a^i_{\ell k}\big)\cr 
&\ge \sum_{j\in A} \Big(f_{ij}\big(\!\!\!\sum_{k\in B\cup\{\ell\}}\!\!\!a^i_{jk}\big)-f_{ij}\big(\sum_{k\in B}a^i_{jk}\big)\Big)+f_{i\ell}\big(\!\!\!\sum_{k\in A\cup\{\ell\}}\!\!\!a^i_{\ell k}\big)\cr
&\ge \sum_{j\in A} \Big(f_{ij}\big(\!\!\!\sum_{k\in A\cup\{\ell\}}\!\!\!a^i_{jk}\big)-f_{ij}\big(\sum_{k\in A}a^i_{jk}\big)\Big)+f_{i\ell}\big(\!\!\!\sum_{k\in A\cup\{\ell\}}\!\!\!a^i_{\ell k}\big)\cr
&=f_i(A\cup\{\ell\})-f_i(A),
\end{align} 
where the first inequality holds by the monotonicity of functions $f_{ij}$ and by $a^i_{jk}\ge 0, A\subseteq B$, and the second inequality follows from convexity of the functions $f_{ij}$. 
\end{proof}

Next, let us consider the IP formulation \eqref{eq:SM} for the MSW, where $f_i:\{0,1\}^n\to \mathbb{R}_+$ is given by $f_i(\boldsymbol{x}_i)=\sum_jf_{ij}\big(\sum_{k}a^i_{kj}x_{ki}x_{ji}\big)$. As each $f_{ij}$ is a convex and nondecreasing function, using Lemma \ref{lemm:convex-supper}, each $f_i$ is a monotone nonnegative supermodular function. Now let $f^L_i(\boldsymbol{x}_i):[0,1]^n\to \mathbb{R}_+$ be the Lov\'asz extension of $f_i(\boldsymbol{x}_i)$ given by 
\begin{align}\label{eq:theta-polynomial-f}
f^L_i(\boldsymbol{x}_i)\!=\!\mathbb{E}_{\theta}[f_i(\boldsymbol{x}^{\theta}_i)]\!=\!\mathbb{E}_{\theta}\Big[\sum_jf_{ij}\big(\sum_{k}a^i_{kj}x^{\theta}_{ki}x^{\theta}_{ji}\big)\Big].
\end{align} 
Since $f_i(\boldsymbol{x}_i)$ is supermodular, the function $f_i^{L}(\boldsymbol{x}_i)$ is a nonnegative concave function. As the objective function in \eqref{eq:SM} is separable across variables $\boldsymbol{x}_i, i\in [m]$, by linearity of expectation $\sum_{i=1}^m f^L_i(\boldsymbol{x}_i)$ equals to the Lov\'asz extension of the objective function in \eqref{eq:SM}, which is also a concave function. Therefore, we obtain the following concave relaxation for the IP \eqref{eq:SM} whose optimal value upper-bounds that of \eqref{eq:SM}. 
 \begin{align}\label{eq:convex-Rel}
 \max \Big\{\sum_{i=1}^m f^L_i(\boldsymbol{x}_i): \ \sum_{i=1}^m\boldsymbol{x}_i\leq \boldsymbol{1}, \ \boldsymbol{x}_i\ge \boldsymbol{0}, \forall i\in [m]\Big\}.
 \end{align}
 
\subsection{Positive Polynomial Externalities of Bounded Degree}

Here, we consider convex externality functions that can be represented by polynomials of the form $f_{ij}(y)=\sum_{r\in [d]}c_{r-1}y^{r-1}$ with nonnegative coefficients $c_{r-1}\ge 0, \forall r\in [d]$. In particular, we show that a slight variant of the randomized rounding algorithm derived from the work of Kleinberg and Tardos (KT) for metric labeling \cite{kleinberg2002approximation} provides a $d$-approximation for the IP \eqref{eq:SM} when applied to the optimal solution of the concave program \eqref{eq:convex-Rel}. The rounding scheme is summarized in Algorithm \ref{alg-main}. The algorithm proceeds in several rounds until all the agents are assigned an item. At each round, the algorithm selects a random item $I\in [m]$ and a random subset of unassigned agents $S_I^{\theta}\subseteq [n]\setminus S$, and assign item $I$ to the agents in the set $S_I^{\theta}$.

\begin{algorithm}[h]\caption{Iterative KT Rounding Algorithm}\label{alg-main}
$\bullet$ Let $\boldsymbol{x}$ be the optimal solution to the concave program \eqref{eq:convex-Rel}.

\noindent 
$\bullet$ During the course of the algorithm, let $S$ be the set of allocated agents and $S_i$ be the set of agents that are allocated item $i$. Initially set $S=\emptyset$ and $S_i=\emptyset, \forall i$.

\noindent
$\bullet$ While $S\neq [n]$, pick $i\in [m], \theta \in [0, 1]$ uniformly at random. Let $S_i^{\theta}:=\{j\in [n]\setminus S: x_{ji}\ge \theta\}$, and update $S_i\leftarrow S_i\cup S_i^{\theta}$ and $S\leftarrow S\cup S_i^{\theta}$.

\noindent
$\bullet$ Return $S_1, \ldots, S_m$. 
\end{algorithm}

In the following lemma, we show that if the externality functions $f_{ij}$ can be represented (or uniformly approximated) by nonnegative-coefficient polynomials of degree less than $d$, then the expected utility of the agents assigned during the first round of Algorithm \ref{alg-main} is at least $\frac{1}{d}$ of the expected value that those agents fractionally contribute to the Lov\'asz extension objective function.

\begin{lemma}\label{lemm:increment}
Assume each externality function $f_{ij}$ is a polynomial with nonnegative coefficients of degree less than $d$. Let $f^{L}(\boldsymbol{x})=\sum_{i=1}^{m}f_i^{L}(\boldsymbol{x}_i)$, where $\boldsymbol{x}$ is an $m\times n$ feasible solution to \eqref{eq:convex-Rel} whose $i$th column equals $\boldsymbol{x}_i$. Moreover, let $A=S_I^{\theta}$ be the random (possibly empty) set of agents that are selected during the \emph{first} round of Algorithm \ref{alg-main}. Then,
\begin{align}\nonumber
\mathbb{E}[f_I(A)]\ge \frac{1}{d}\mathbb{E}[f^L(\boldsymbol{x})-f^L(\boldsymbol{x}_{|_{\bar{A}}})],
\end{align}
where $f^L(\boldsymbol{x}_{|_{\bar{A}}})$ denotes the value of the Lov\'asz extension $f^L(\cdot)$ when its argument is restricted to the rows of $\boldsymbol{x}$ corresponding to the agents $j\in \bar{A}=[n]\setminus A$.\footnote{Equivalently, $f^L(\boldsymbol{x}_{|_{\bar{A}}})$ is equal to evaluating $f^L$ at a solution that is obtained from $\boldsymbol{x}$ by setting all the rows corresponding to agents in $A$ to $\boldsymbol{0}$.}
\end{lemma}
\begin{proof}
First, we note that
\begin{align}\label{eq:f-I-A}
\mathbb{E}[f_I(A)]&=\frac{1}{m}\sum_{i=1}^{m}\mathbb{E}_{\theta}[f_{i}(S_i^{\theta})]\cr 
&=\frac{1}{m}\sum_{i=1}^{m}\mathbb{E}_{\theta}[f_{i}(\boldsymbol{x}_i^{\theta})]\cr 
&=\frac{1}{m}f^L(\boldsymbol{x}),
\end{align}
where we recall that $x^{\theta}_{ji}=1$ if and only if $j\in S_i^{\theta}$. Since each $x_{ji}^{\theta}$ is a binary random variable, we have $(x_{ji}^{\theta})^t=x_{ji}^{\theta}, \forall t\ge 1$. As each $f_{ij}$ is a polynomial with nonnegative coefficients of degree less than $d$, after expanding all the terms in \eqref{eq:theta-polynomial-f}, there are nonnegative coefficients $b^{i}_{j_1,\ldots,j_r}$ such that
\begin{align}\nonumber
f_i(\boldsymbol{x}^{\theta}_i)&=\sum_jf_{ij}\big(\!\sum_{k}a^i_{jk}x^{\theta}_{ki}x^{\theta}_{ji}\big)\cr
&=\sum_{r=2}^d\sum_{j_1,\ldots,j_{r}}\!b^{i}_{j_1,\ldots,j_r} \!\prod_{\ell=1}^rx^{\theta}_{j_{\ell}i}.
\end{align}   
Note that we may assume $f_{ij}$ does not have any constant term as it does not affect the MSW optimization. Taking expectation from the above relation, we obtain
\begin{align}\nonumber
f^L_i(\boldsymbol{x}_i)&=\sum_{r=2}^d\sum_{j_1,\ldots,j_{r}}b^{i}_{j_1,\ldots,j_r} \mathbb{E}_{\theta}\big[\prod_{\ell=1}^rx^{\theta}_{j_{\ell}i}\big]\cr 
&=\sum_{r=2}^d\sum_{j_1,\ldots,j_{r}}b^{i}_{j_1,\ldots,j_r} \mathbb{P}_{\theta}\big(x^{\theta}_{j_{\ell}i}=1, \forall \ell\in [r]\big)\cr 
&=\sum_{r=2}^d\sum_{j_1,\ldots,j_{r}}b^{i}_{j_1,\ldots,j_r} \mathbb{P}\big(\theta\leq x_{j_{\ell}i}, \forall \ell\in [r]\big)\cr 
&=\sum_{r=2}^d\sum_{j_1,\ldots,j_{r}}b^{i}_{j_1,\ldots,j_r} \min_{\ell\in [r]}\{x_{j_{\ell}i}\}.
\end{align}
For any $\ell\in [r]$, with some abuse of notation, let $\boldsymbol{x}_{j_{\ell}}=(x_{j_{\ell}i}, i\in [m])$ denote the $j_{\ell}$-th \emph{row} of the solution $\boldsymbol{x}$, and define $f^L(\boldsymbol{x}_{j_1},\ldots,\boldsymbol{x}_{j_r})=\sum_i b^{i}_{j_1,\ldots,j_r}\min_{\ell\in [r]}\{x_{j_{\ell}i}\}$ to be the restriction of $f^{L}$ to the rows $\boldsymbol{x}_{j_{\ell}}, \ell\in [r]$. Using the above expression we have
\begin{align}\label{eq:f-r-extend}
f^L(\boldsymbol{x})=\sum_{i=1}^mf_i^L(\boldsymbol{x}_i)=\sum_{r=2}^d\sum_{j_1,\ldots,j_{r}}f^L(\boldsymbol{x}_{j_1},\ldots,\boldsymbol{x}_{j_r}).
\end{align}
Using \eqref{eq:f-r-extend}, we note that a tuple of rows $\boldsymbol{x}_{j_1},\ldots,\boldsymbol{x}_{j_{r}}$ contribute exactly $f^L(\boldsymbol{x}_{j_1},\ldots,\boldsymbol{x}_{j_r})$ to the objective $f^L(\boldsymbol{x})-f^L(\boldsymbol{x}_{|_{\bar{S}^{\theta}_i}})$ if at least one of the agents $j_{\ell}, \ell\in [r]$ belong to $S_i^{\theta}$, and contribute $0$, otherwise. Therefore, if $A=S^{\theta}_I$ is the random set obtained during the first round of Algorithm \ref{alg-main},  using linearity of expectation, we can write
\begin{align}\label{eq:max-2}
&\mathbb{E}[f^L(\boldsymbol{x})-f^L(\boldsymbol{x}_{|_{\bar{A}}})]=\frac{1}{m}\sum_{i=1}^{m}\mathbb{E}[f^L(\boldsymbol{x})-f^L(\boldsymbol{x}_{|_{\bar{S}^{\theta}_i}})]\cr 
&=\frac{1}{m}\sum_{i=1}^m\sum_{r=2}^{d}\sum_{j_1,\ldots,j_r}\mathbb{P}\big(\cup_{\ell=1}^{r}\{j_{\ell}\in S^{\theta}_i\}\big)f^{L}(\boldsymbol{x}_{j_1},\ldots,\boldsymbol{x}_{j_r})\cr 
&=\frac{1}{m}\sum_{i=1}^m\sum_{r=2}^{d}\sum_{j_1,\ldots,j_r}\max_{\ell\in [r]}\{x_{j_{\ell}i}\}f^{L}(\boldsymbol{x}_{j_1},\ldots,\boldsymbol{x}_{j_r})\cr
&=\frac{1}{m}\sum_{r=2}^{d}\sum_{j_1,\ldots,j_r}\!\!\!\Big(\!\sum_{i=1}^m\max_{\ell\in [r]}\{x_{j_{\ell}i}\}\!\Big)f^{L}(\boldsymbol{x}_{j_1},\ldots,\boldsymbol{x}_{j_r})\cr
&\leq \frac{1}{m}\sum_{r=2}^{d}\sum_{j_1,\ldots,j_r}rf^{L}(\boldsymbol{x}_{j_1},\ldots,\boldsymbol{x}_{j_r})\cr
&\leq \frac{d}{m}\sum_{r=2}^{d}\sum_{j_1,\ldots,j_r}f^{L}(\boldsymbol{x}_{j_1},\ldots,\boldsymbol{x}_{j_r})=\frac{d}{m}f^{L}(\boldsymbol{x}),
\end{align}
where the first inequality holds because by feasibility of the solution $\boldsymbol{x}$, we have $\sum_{i}\max_{\ell\in [r]}\{x_{j_{\ell}i}\}\leq \sum_{i}\sum_{\ell=1}^{r}x_{j_{\ell}i}\leq r$, and the second inequality holds because the terms $f^{L}(\boldsymbol{x}_{j_1},\ldots,\boldsymbol{x}_{j_r})$ are nonnegative. Combining relations \eqref{eq:f-I-A} and \eqref{eq:max-2} completes the proof. 
\end{proof}


\begin{theorem}\label{thm:base-2}
Assume each externality function $f_{ij}$ is a polynomial with nonnegative coefficients of degree less than $d$. Then Algorithm \ref{alg-main} is a $d$-approximation algorithm for the MSW \eqref{eq:general-welfare-convex}.
\end{theorem} 
\begin{proof}
We use an induction on the number of agents to show that the expected value of the solution returned by Algorithm \ref{alg-main} is at least $\frac{1}{d}f^{L}(\boldsymbol{x})$, where $f^{L}(\boldsymbol{x})=\sum_{i=1}^{m}f_i^{L}(\boldsymbol{x}_i)$, and $\boldsymbol{x}$ is the \emph{optimal} solution to \eqref{eq:convex-Rel}. Without loss of generality, we may assume that the random set $A$ that is selected during the first round of Algorithm \ref{alg-main} is nonempty. Otherwise, no update occurs, and we can focus on the first iterate in which a nonempty set is selected.

The base case in which there is only $n=1$ agent follows trivially from Lemma \ref{lemm:increment}, because nonemptyness of $A$ implies $\bar{A}=\emptyset$, and thus $\mathbb{E}[{\rm Alg}]=\mathbb{E}[f_I(A)]\ge \frac{1}{d}f^{L}(\boldsymbol{x})$. Now assume that the induction hypothesis holds for any set of at most $n-1$ agents. Given an instance with $n$ agents, let $A=S_{I}^{\theta}$ be the nonempty random set of agents that are selected during the first round of Algorithm \ref{alg-main}. Moreover, let $S_1,\ldots,S_m$ be the (random) sets returned by the algorithm when applied on the \emph{remaining} agents in $\bar{A}$. As $|\bar{A}|\leq n-1$, using induction hypothesis on the agents $\bar{A}$, we have
\begin{align}\label{eq:induction}
\mathbb{E}[\sum_{i}f_i(S_i)|A]\ge \frac{1}{d}\max_{\boldsymbol{y}\boldsymbol{1}\leq \boldsymbol{1}, \boldsymbol{y}\ge 0} f^L(\boldsymbol{y})\ge \frac{1}{d} f^L(\boldsymbol{x}_{|_{\bar{A}}}),
\end{align}
where $\boldsymbol{y}\in \mathbb{R}^{|\bar{A}|\times m}_+$ is a variable restricted only to the agents in $\bar{A}$, and the second inequality holds because $\boldsymbol{x}_{|_{\bar{A}}}$ is a feasible solution to the middle maximization. (Recall that $\boldsymbol{x}_{|_{\bar{A}}}$ is the portion of solution $\boldsymbol{x}$ when restricted to rows $j\in \bar{A}$.) Now, we have
\begin{align}\nonumber
\mathbb{E}[{\rm Alg}]&=\mathbb{E}[\sum_{i\neq I}f_i(S_i)+f_I(S_I\cup A)]\cr 
&\ge \mathbb{E}[\sum_{i}f_i(S_i)+f_I(A)]\cr 
&= \mathbb{E}\big[\mathbb{E}[\sum_{i}f_i(S_i)+f_I(A)|A]\big]\cr 
&=\mathbb{E}\big[\mathbb{E}[\sum_{i}f_i(S_i)|A]+f_I(A)\big]\cr 
&\ge \frac{1}{d} \mathbb{E}[f^L(\boldsymbol{x}_{|_{\bar{A}}})]+\frac{1}{d}\mathbb{E}[f^L(\boldsymbol{x})-f^L(\boldsymbol{x}_{|_{\bar{A}}})]\cr 
&=\frac{1}{d} f^L(\boldsymbol{x}),
\end{align}
where the first inequality uses the superadditivity of $f_i$ due to supermodular property (that is $f_i(P\cup Q)\ge f_i(P)+f_i(Q)$ for any $P\cap Q=\emptyset$), and the last inequality holds by \eqref{eq:induction} and Lemma \ref{lemm:increment}. 
\end{proof}

\begin{corollary}
For the special case of positive linear externalities $f_{ij}(y)=y$, $\forall i,j$, one can take $d=2$, in which case Algorithm \ref{alg-main} is a $2$-approximation algorithm. Interestingly, derandomization of Algorithm \ref{alg-main} in this special case recovers the iterative greedy algorithm developed in \cite[Theorem 4]{de2012finding}, which first solves a linear program relaxation for MSW and then rounds the solution using an iterative greedy algorithm.
\end{corollary}

\begin{remark}
For convex polynomial externalities of degree at most $d$, the $d$-approximation guarantee of Theorem \ref{thm:base-2} is an exponential improvement over the $2^{O(d)}$-approximation guarantee given in \cite[Theorem 3.9]{bhalgat2012mechanisms}. 
\end{remark}

\subsection{Monotone Convex Externalities of Bounded Curvature}

In this part, we provide an approximation algorithm for the MSW with general monotone and positive convex externalities. Unfortunately, for general convex externalities, the Lov\'asz extension of the objective function does not admit a closed-form structure. For that reason, we develop an approximation algorithm whose performance guarantee depends on the curvature of the externality functions.

\begin{definition}\label{def:curvature}
Given $\alpha\in (0, 1)$, we define the $\alpha$-curvature of a nonnegative nondecreasing convex function $h:[0, b]\to \mathbb{R}_+$ as $\gamma_{\alpha}^h:=\inf_{y\in (0, b]}\frac{h(\alpha y)}{h(y)}$.
\end{definition}

\begin{remark}
Using monotonicity of $h$, we always have $\gamma_{\alpha}^{h}\in [0, 1]$. In particular, for any monotone $k$-homogeneous convex function $h(\alpha y)\ge \alpha^{k}h(y)$, we have $\gamma_{\alpha}^h\ge \alpha^{k}$.
\end{remark}

It is worth noting that \cite{conforti1984submodular} also develops a curvature-dependent greedy approximation algorithm for maximizing a nondecreasing \emph{submodular} function subject to a matroid constraint. However, the definition of curvature in \cite{conforti1984submodular} is different from ours as it looks at the maximum normalized growth rate of the overall objective function $f$ as a new element is added to the solution set. Moreover, here we are looking at \emph{supermodular} maximization (or submodular minimization) that behaves completely different in terms of approximability and solution method. In fact, for the case of submodular maximization, Theorem \ref{thm:multi-negative} already provides a curvature-independent $e$-approximation algorithm.   


Using Lemma \ref{lemm:convex-supper} the MSW \eqref{eq:general-welfare-convex} with monotone convex externality functions can be cast as the supermodular maximization problem \eqref{eq:SM}. Relaxing that problem via Lov\'asz extension, we obtain the concave program \eqref{eq:convex-Rel}, whose optimal solution, denoted by $\boldsymbol{x}$, can be found in polynomial time. We round the optimal fractional solution $\boldsymbol{x}$ to an integral one $\hat{\boldsymbol{x}}$ using the two-stage fair contention resolution scheme. It is instructive to think about the rounding process as a two-stage process for rounding the fractional $n\times m$ matrix $\boldsymbol{x}$. In the first stage, the columns are rounded independently, and in the second stage, the rows of the resulting solution are randomly revised to create the final integral solution $\hat{\boldsymbol{x}}$, satisfying the partition constraints in \eqref{eq:SM}. The rounding algorithm is summarized in Algorithm \ref{alg:contention}.

\begin{algorithm}[h]\caption{Rounding for Positive Convex Externalities}\label{alg:contention}
$\bullet$ {\bf Input:} The optimal $n\times m$ solution $\boldsymbol{x}$ to the concave program \eqref{eq:convex-Rel}.

\noindent 
$\bullet$ {\bf Stage 1:} For each $i\in [m]$ pick an independent uniform random variable $\theta_i\in [0, 1]$, and for $j\in[n]$ let $x^{\theta_i}_{ji}=1$ if $x_{ji}\ge \theta_i$, and $x^{\theta_i}_{ji}=0$, otherwise. Let $\boldsymbol{x}^{\boldsymbol{\theta}}=[\boldsymbol{x}_1^{\theta_1}|\ldots|\boldsymbol{x}_m^{\theta_m}]$ be the binary random matrix obtained at the end of stage 1.

\noindent
$\bullet$ {\bf Stage 2:}  For each $j\in [n]$, let $A_j:=\{i: x^{\theta_i}_{ji}=1\}$ be a random set denoting the positions in the $j$-th row of $\boldsymbol{x}^{\boldsymbol{\theta}}$ that are rounded to $1$ in the first stage. If $|A_j|\leq 1$, do nothing. Otherwise, set all the entries of the $j$th row of $\boldsymbol{x}^{\boldsymbol{\theta}}$ to $0$ except the $i$th entry, where $i$ is selected from $A_j$ with probability $r_{iA_j}=\sum_{i'\in A_j\setminus \{i\}}\frac{x_{ji'}}{|A_j|-1}+\sum_{i'\notin A_j}\frac{x_{ji'}}{|A_j|}$.

\noindent
$\bullet$ {\bf Output:} $\hat{\boldsymbol{x}}$. 
\end{algorithm}

\begin{theorem}
Algorithm \ref{alg:contention} is a randomized $\gamma_{\frac{1}{4}}^{-1}$-approximation algorithm for the MSW \eqref{eq:general-welfare-convex} with positive convex externality functions $f_{ij}$, where $\gamma_{\frac{1}{4}}:=\min_{i,j}\gamma_{\frac{1}{4}}^{f_{ij}}$.   
\end{theorem}
\begin{proof}
The first stage in Algorithm \ref{alg:contention} simply follows from the definition of the Lov\'asz extension such that $\sum_i\mathbb{E}[f_i(\boldsymbol{x}^{\theta_i}_i)]=\sum_if^{L}_i(\boldsymbol{x}_i)$. Thus, after the first stage we obtain a binary random matrix $\boldsymbol{x}^{\boldsymbol{\theta}}=[\boldsymbol{x}_1^{\theta_1}|\ldots|\boldsymbol{x}_m^{\theta_m}]$ whose expected objective value equals to the optimal value of the concave relaxation \eqref{eq:convex-Rel}. Unfortunately, after the first phase, the rounded solution may not satisfy the partition constraints as multiple items may be assigned to the same agent $j$. The second stage resolves that issue by modifying $\boldsymbol{x}^{\boldsymbol{\theta}}$ to $\hat{\boldsymbol{x}}$ by separately applying the fair contention resolution to every row of $\boldsymbol{x}^{\boldsymbol{\theta}}$. Thus, after the second stage of rounding, $\hat{\boldsymbol{x}}$ is a feasible solution to \eqref{eq:SM}.  

Since $\theta_i, i\in [m]$ are independent uniform random variables, for any $j\in [n]$, $\mathbb{P}\{x^{\theta_i}_{ji}=1\}=\mathbb{P}\{\theta_i\leq x_{ji}\}=x_{ji}$. Therefore, for any row $j$, one can imagine that the $m$ entries of row $j$ compete independently with probabilities $\{x_{ji}\}_{i\in [m]}$ to receive the resource in the contention resolution scheme. Let us consider an arbitrary column $i$. Using Lemma \ref{lemm:contention} for row $j$, and given $i\in A_j$, the probability that item $j$ is given to player $i$ is at least  $1-\frac{1}{e}$, that is $\mathbb{P}\{\hat{x}_{ji}=1|x^{\theta_i}_{ji}=1\}\ge 1-\frac{1}{e}, \forall j\in [n]$. Now consider any $k\neq j$ and note that $\mathbb{P}\{\hat{x}_{ji}=1|x^{\theta_i}_{ji}=1, x^{\theta_i}_{ki}=1\}=\mathbb{P}\{\hat{x}_{ji}=1|x^{\theta_i}_{ji}=1\}$, as the event $\{\hat{x}_{ji}=1|x^{\theta_i}_{ji}=1\}$ is independent of the event $\{x^{\theta_i}_{ki}=1\}$. Using union bound, we can write
\begin{align}\nonumber
&\mathbb{P}\{\hat{x}_{ji}=1, \hat{x}_{ki}=1|x^{\theta_i}_{ji}=1, x^{\theta_i}_{ki}=1\}\cr 
&\qquad=1-\mathbb{P}\{\hat{x}_{ji}=0 \cup \hat{x}_{ki}=0|x^{\theta_i}_{ji}=1, x^{\theta_i}_{ki}=1\}\cr 
&\qquad\ge 1-\mathbb{P}\{\hat{x}_{ji}=0 |x^{\theta_i}_{ji}=1, x^{\theta_i}_{ki}=1\}\cr 
&\qquad\qquad-\mathbb{P}\{\hat{x}_{ki}=0|x^{\theta_i}_{ji}=1, x^{\theta_i}_{ki}=1\}\cr 
&\qquad=\mathbb{P}\{\hat{x}_{ji}=1 |x^{\theta_i}_{ji}=1, x^{\theta_i}_{ki}=1\}\cr 
&\qquad\qquad+\mathbb{P}\{\hat{x}_{ki}=1|x^{\theta_i}_{ji}=1, x^{\theta_i}_{ki}=1\}-1\cr 
&\qquad=\mathbb{P}\{\hat{x}_{ji}=1 |x^{\theta_i}_{ji}=1\}\!+\!\mathbb{P}\{\hat{x}_{ki}=1|x^{\theta_i}_{ki}=1\}-1> \frac{1}{4}.
\end{align}
Using Jensen's inequality, we can lower-bound the expected objective value of $\hat{\boldsymbol{x}}$ as
\begin{align}\label{eq:jensen}
&\sum_{i,j}\mathbb{E}\big[f_{ij}\big(\sum_{k}a^i_{jk}\hat{x}_{ki}\hat{x}_{ji}\big)\big]\cr 
&\qquad=\sum_{i,j}\mathbb{E}_{\theta_i}\Big[\mathbb{E}\big[f_{ij}\big(\sum_{k}a^i_{jk}\hat{x}_{ki}\hat{x}_{ji}\big)|\boldsymbol{x}_i^{\theta_i}\big]\Big]\cr
&\qquad\ge \sum_{i,j}\mathbb{E}_{\theta_i}\Big[f_{ij}\big(\mathbb{E}\big[\sum_{k}a^i_{jk}\hat{x}_{ki}\hat{x}_{ji}|\boldsymbol{x}_i^{\theta_i}\big]\big)\Big]\cr
&\qquad= \sum_{i,j}\mathbb{E}_{\theta_i}\Big[f_{ij}\big(\sum_{k}a^i_{jk}\mathbb{E}[\hat{x}_{ki}\hat{x}_{ji}|\boldsymbol{x}_i^{\theta_i}]\big)\Big],
\end{align}
where the inner expectation in the first equality is with respect to $\boldsymbol{\theta}_{-i}=(\theta_{i'}, i'\neq i)$ and the randomness introduced by the contention resolution in the second phase. Let $\boldsymbol{1}_{\{\cdot\}}$ denote the indicator function. Then, for any $i,j,k$, we have
\begin{align}\label{eq:final-convex-r}
&\mathbb{E}[\hat{x}_{ki}\hat{x}_{ji}|\boldsymbol{x}_i^{\theta_i}]=\mathbb{E}[\hat{x}_{ki}\hat{x}_{ji}|x_{ji}^{\theta_i},x_{ki}^{\theta_i}]\cr 
&=\mathbb{E}[\hat{x}_{ki}\hat{x}_{ji}|x_{ji}^{\theta_i}=1,x_{ki}^{\theta_i}=1]\cdot \boldsymbol{1}_{\{x_{ij}^{\theta_i}=1,x_{ik}^{\theta_i}=1\}}\cr 
&=\mathbb{P}\{\hat{x}_{ki}=1, \hat{x}_{ji}=1 |x_{ji}^{\theta_i}=1,x_{ki}^{\theta_i}=1\}\cdot \boldsymbol{1}_{\{x_{ji}^{\theta_i}=1,x_{ki}^{\theta_i}=1\}}\cr 
&\ge \frac{1}{4}\cdot \boldsymbol{1}_{\{x_{ji}^{\theta_i}=1,x_{ki}^{\theta_i}=1\}}=\frac{1}{4}x_{ji}^{\theta_i}x_{ki}^{\theta_i}.
\end{align}  
Substituting the above relation into \eqref{eq:jensen}, and using the monotonicity of $f_{ij}$ together with Definition \ref{def:curvature}, we can write
\begin{align}\nonumber
\mathbb{E}\Big[\sum_{i,j}f_{ij}&\big(\sum_{k}a^i_{jk}\hat{x}_{ki}\hat{x}_{ji}\big)\Big]\!\ge\! \sum_{i,j}\mathbb{E}_{\theta_i}\Big[f_{ij}\big(\sum_{k}\frac{a^i_{jk}}{4}x_{ji}^{\theta_i}x_{ki}^{\theta_i}\big)\Big]\cr 
&\ge \gamma_{\frac{1}{4}} \sum_{i}\mathbb{E}_{\theta_i}\Big[\sum_j f_{ij}\big(\sum_{k}a^i_{jk}x_{ji}^{\theta_i}x_{ki}^{\theta_i}\big)\Big]\cr 
&=\gamma_{\frac{1}{4}} \sum_{i}\mathbb{E}_{\theta_i}\Big[f_i(\boldsymbol{x}^{\theta_i}_i)\Big]=\gamma_{\frac{1}{4}} \sum_{i}f_i^L(\boldsymbol{x}_i).
\end{align}
Therefore, the expected value of the rounded solution is at least $\gamma_{\frac{1}{4}}$ times the optimal Lov\'asz relaxation, which completes the proof. 
\end{proof}

\subsection{A Numerical Example}

Here, we provide a numerical experiment to verify the performance guarantee of the algorithms developed in this section. In our numerical experiment, we fix the number of items to $m=10$, and the externality functions to be linear $f_{ij}(y)=y, \forall i,j$. As a result, the objective function for MSW can be written as $f(\boldsymbol{x})=\sum_{i=1}^{10}\boldsymbol{x}_i'\boldsymbol{A}_i\boldsymbol{x_i}$. 

We generate $40$ different instances as the number of agents increase from $n=10$ to $n=50$. Given an instance with $n$ agents, we generate the weight matrices $\boldsymbol{A}_i\in \{0,1\}^{n\times n}$ by randomly selecting $10$ rows in $\boldsymbol{A}_i$ and uniformly setting one of the elements in that row to $1$, and the remaining elements of that row to $0$. The expected objective value of Algorithm \ref{alg-main}, Algorithm \ref{alg:contention}, and the optimal IP \eqref{eq:SM} are illustrated in Figure \ref{fig:flow}, where the $x$-axis corresponds to different instances of $n=10,\ldots,50$, and the $y$-axis shows the expected objective value. While in this specific example Algorithm \ref{alg-main} mostly outperforms Algorithm \ref{alg:contention}, however, as can be seen, the expected objective value of both algorithms is close to the optimal IP objective value. In particular, for all the instances, Algorithm \ref{alg-main} achieves at least $\frac{1}{2}$ of the optimal objective value. 

\begin{figure}[t]
\vspace{-2.7cm}
\begin{center}
\includegraphics[totalheight=.3\textheight,
width=.45\textwidth,viewport=-20 0 330 350]{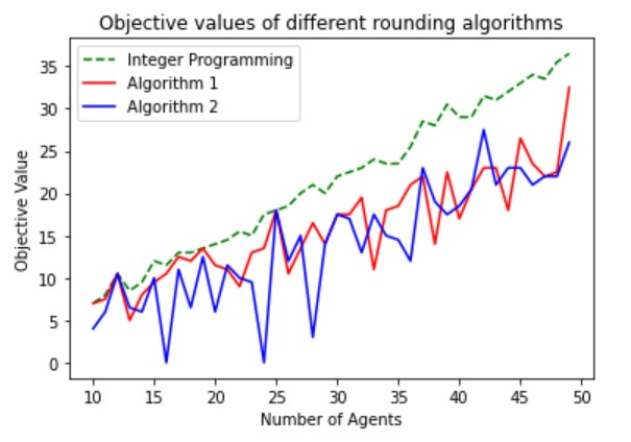}
\end{center}
\vspace{-0.3cm}
\caption{Illustration of the performance of Algorithms \ref{alg-main} and \ref{alg:contention} for positive linear externalities.}\label{fig:flow}\vspace{-0.3cm}
\end{figure}

\section{MSW with Positive Monotone Concave Externalities}\label{sec:concave}
In this section, we extend our results to approximate MSW with nondecreasing positive concave externality functions. Unfortunately, for positive concave externalities, Lemmas \ref{lemm:convex-supper} and \ref{lemm:convex-sub} do not hold, and the objective function in \eqref{eq:SM} is no longer supermodular or submodular. For that reason, we cannot directly use the Lov\'asz or multilinear extensions to solve or approximate the continuous relaxation of the MSW. To address this issue, in this section, we take two different approaches based on a combination of the ideas that have been developed so far. Each method is suitable for a particular subclass of concave functions and together provides a good understanding of how to solve MSW with positive concave externalities.

\subsection{Positive Concave Externalities of Small Curvature}

For simplicity and without any loss of generality, throughout this section we assume that the influence weights are normalized such that $\sum_{k=1}^na^i_{jk}=1, \forall i,j$. Otherwise, we can scale and redefine the externality functions as $f_{ij}(y)\leftarrow f_{ij}((\sum_{k=1}^na^i_{jk})y)$. Note that such a scaling preserves concavity and monotonicity, and for the new externalities we have $f_{ij}:[0,1]\to \mathbb{R}_+$. The following proposition provides a performance guarantee for approximating MSW with positive concave externalities, which is particularly effective for concave externalities of small curvature.

\begin{proposition}
For nondecreasing concave externalities $f_{ij}:[0,1]\to \mathbb{R}_+$, let $\beta:=\max_{ij}\sup_{X}\frac{f_{ij}(\mathbb{E}[X])}{\mathbb{E}[f_{ij}(X)]}$, where the $\sup$ is over all random variables $X\in [0, 1]$. Then, the MSW with positive concave externalities admits a $4\beta$-approximation algorithm.
\end{proposition}
\begin{proof}
Let us consider the IP \eqref{eq:IP-convex} for the MSW with positive concave externality functions $f_{ij}$, and note that $x_{ji}x_{ki}=\min\{x_{ji},x_{ki}\}$ for any two binary variables $x_{ji},x_{ki}\in \{0,1\}$. By replacing this relation into the objective function of IP \eqref{eq:IP-convex} and relaxing the binary constraints, we obtain the following concave relaxation for MSW:
\begin{align}\label{eq:concave-relaxation-exter}
&\max \sum_{i,j}f_{ij}\Big(\sum_{k=1}^na^i_{jk}\min\{x_{ji},x_{ki}\}\Big)\cr 
&\qquad \sum_{i=1}^m x_{ji}\leq 1\ \forall j,\cr
&\qquad  x_{ji}\ge 0\  \forall i,j,
\end{align}
where the concavity of the objective function follows from the concavity of $f_{ij}$ and the concavity of $\sum_{k=1}^na^i_{jk}\min\{x_{ji},x_{ki}\}$. Therefore, one can solve \eqref{eq:concave-relaxation-exter} in polynomial time to obtain an optimal fractional solution $\boldsymbol{x}$. Using this solution as an input to Algorithm \ref{alg:contention} we obtain a feasible integral solution $\hat{\boldsymbol{x}}$, whose expected objective value can be lower-bounded as
\begin{align}\nonumber
&\sum_{i,j}\mathbb{E}\big[f_{ij}\big(\sum_{k}a^i_{jk}\hat{x}_{ki}\hat{x}_{ji}\big)\big] 
\!\ge\! \frac{1}{\beta}\sum_{i,j}f_{ij}\big(\sum_{k}a^i_{jk}\mathbb{E}[\hat{x}_{ki}\hat{x}_{ji}]\big)\cr
&\qquad\qquad= \frac{1}{\beta}\sum_{i,j}f_{ij}\big(\sum_{k}a^i_{jk}\mathbb{E}_{\theta_i}\big[\mathbb{E}[\hat{x}_{ki}\hat{x}_{ji}|\boldsymbol{x}_i^{\theta_i}]\big]\big)\cr
&\qquad\qquad\ge \frac{1}{\beta}\sum_{i,j}f_{ij}\big(\sum_{k}\frac{a^i_{jk}}{4}\mathbb{E}_{\theta_i}[x_{ji}^{\theta_i}x_{ki}^{\theta_i}]\big)\cr 
&\qquad\qquad= \frac{1}{\beta} \sum_{i,j} f_{ij}\big(\sum_{k}\frac{a^i_{jk}}{4}\min\{x_{ji},x_{ki}\}\big)\cr 
&\qquad\qquad\ge \frac{1}{4\beta}\sum_{i,j} f_{ij}\big(\sum_{k}a^i_{jk}\min\{x_{ji},x_{ki}\}\big),
\end{align}
where the first inequality uses the definition of $\beta$, the second inequality uses \eqref{eq:final-convex-r}, and the last inequality follows from concavity of $f_{ij}$ and the fact that $f_{ij}(0)=0$.
\end{proof}

\subsection{Multilinear Extension for Positive Concave Externalities}
This final section provides an alternative approach based on the multilinear extension to solve the MSW subject to positive concave externalities approximately. Let us again consider the IP formulation \eqref{eq:IP-convex}. By defining new binary variables $y^i_{jk}=x_{ji}x_{ki}$, we can rewrite IP \eqref{eq:IP-convex} as
\begin{align}\label{eq:sub-y}
&\max f(\boldsymbol{y}):=\sum_{i,j}f_{ij}\Big(\sum_{k=1}^n a^i_{jk}y^i_{jk}\Big)\cr 
&\qquad \sum_{i=1}^m x_{ji}\leq 1\ \forall j,\cr
&\qquad y^i_{jk}=x_{ji}x_{ki}\ \forall i,j,k,\cr
&\qquad x_{ji}, y^i_{jk}\in\{0,1\}\ \forall i,j,k.
\end{align}
When twice differentiable, a function is submodular if and only if all cross-second-derivatives are non-positive \cite{bach2019submodular}. For fixed $i,j$, and any $\ell_1\neq \ell_2$, we have 
\begin{align}\nonumber
\frac{\partial f_{ij}(\sum_{k=1}^n a^i_{jk}y^i_{jk})}{\partial y^i_{j\ell_1}\partial y^i_{j\ell_2}}=a^i_{j\ell_1}a^i_{j\ell_2}f''_{ij}(\sum_{k=1}^n a^i_{jk}y^i_{jk})\leq 0,
\end{align}
where the second inequality follows by concavity of $f_{ij}$ and because $a^i_{jk}\ge 0\ \forall i,j,k$. Therefore, the objective function in \eqref{eq:sub-y} is a monotone and nonnegative submodular function over the ground set of triples $(i,j,k)\in [m]\times [n]^2$. Unfortunately, the set of constraints in \eqref{eq:sub-y} do not define a matroid. Therefore, to induce a matroid structure on the set of constraints, we relax the second set of constraints in \eqref{eq:sub-y} and replace $\sum_{i=1}^m x_{ji}\leq 1\ \forall j$ by the valid inequalities $\sum_{i=1}^m y^i_{jk}\leq 1\ \forall j,k.$ That gives us the following multilinear relaxation for the MSW:
\begin{align}\label{eq:sub-y-relaxed}
\max \big\{f^M(\boldsymbol{y}): \sum_{i=1}^m y^i_{jk}\leq 1\ \forall j,k,\ \boldsymbol{y}\ge \boldsymbol{0} \big\}.
\end{align} 
Clearly, the set of constraints in \eqref{eq:sub-y-relaxed} defines a partition matroid over $[m]\times [n]^2$, which is a down-monotone solvable polytope. Therefore, we can use Lemma \ref{lemm:submodular-multilinear} to solve the multilinear relaxation \eqref{eq:sub-y-relaxed} within a factor $1-\frac{1}{e}$ and round that fractional solution to a binary vector $\boldsymbol{y}$ (by abuse of notation) using the pipage rounding \cite[Lemma B.3]{vondrak2013symmetry}. Unfortunately, in general, there is no guarantee on whether the rounded solution $\boldsymbol{y}$ can be decomposed into the product of binary variables $x_{ji}x_{ki}$. However, if such a decomposition is possible approximately, we can ensure that $\boldsymbol{y}$ also satisfies the second set of constraints in \eqref{eq:sub-y} approximately. This suggests a greedy algorithm based on the multilinear extension that is summarized in Algorithm \ref{alg:greedy}. 

\begin{algorithm}[H]\caption{Greedy Algorithm for Positive Concave Externalities}\label{alg:greedy}
$\bullet$ Solve the multilinear relaxation \eqref{eq:sub-y-relaxed} using the continuous greedy algorithm (Lemma \ref{lemm:submodular-multilinear}) and round its solution to $\boldsymbol{y}\in \{0,1\}^{mn^2}$ using the pipage rounding \cite{vondrak2013symmetry}.   

\noindent 
$\bullet$ Let $Y\subset [m]\times [n]^2$ be the subset of elements whose characteristic vector is given by the binary vector $\boldsymbol{y}$. Sort the elements in $Y$ according to their marginal contributions to $f(\boldsymbol{y})$, i.e., given $T=\{(i_{\ell},j_{\ell},k_{\ell}), \ell=1,\ldots,r-1\}$, let $(i_{r},j_{r},k_{r})$ be the element in $Y\setminus T$ that maximizes $f(T\cup \{(i_{r},j_{r},k_{r})\})$.

\noindent
$\bullet$ Set $\hat{\boldsymbol{x}}=\boldsymbol{0}$, $\hat{\boldsymbol{y}}=\boldsymbol{0}$, and $S=\emptyset$. Process the elements in $Y$ according to their order as follows: if an element $(i_{r},j_{r},k_{r})$ can be inserted into the solution set $S$ without violating the constraints $\hat{y}^i_{jk}=\hat{x}_{ji}\hat{x}_{ki}\ \forall i,j,k$, then set $S\leftarrow S\cup \{(i_{r},j_{r},k_{r})\}$ and $\hat{y}^{i_r}_{j_rk_r}=1, \hat{x}_{j_ri_r}=1, \hat{x}_{k_ri_r}=1$. Otherwise, skip to the next element $(i_{r+1},j_{r+1},k_{r+1})$.

\noindent
$\bullet$ Return $(\hat{\boldsymbol{y}},\hat{\boldsymbol{x}})$. 
\end{algorithm}

\begin{proposition}
Assume that the number of elements that are skipped during the last round of Algorithm \ref{alg:greedy} is bounded above by $c$. Then, Algorithm \ref{alg:greedy} is an $\frac{e(c+1)}{e-1}$-approximation algorithm for the MSW with positive concave externalities.
\end{proposition}
\emph{Proof:} Let $\hat{\boldsymbol{y}}$ be the characteristic vector of the final solution set $S$ that is generated by Algorithm \ref{alg:greedy}. Clearly, $(\hat{\boldsymbol{y}}, \hat{\boldsymbol{x}})$ is consistent with all the constraints $\hat{y}^i_{jk}=\hat{x}_{ji}\hat{x}_{ki}\ \forall i,j,k$, which means that $(\hat{\boldsymbol{y}}, \hat{\boldsymbol{x}})$ returned by Algorithm \ref{alg:greedy} is a feasible integral solution to the IP \eqref{eq:sub-y}. Moreover, we note that the algorithm always chooses the first element in the sequence because it does not violate any constraint, i.e., $(i_1,j_1,k_1)\in S$. Since the objective function $f(\cdot)$ is a nonnegative nondecreasing submodular function and the elements in $Y$ are processed according to nonincreasing marginal contributions, we have $f(S)\ge f(\{(i_1,j_1,k_1)\})\ge \frac{f(Y\setminus S)}{|Y|-|S|}$. Using this relation and submodularity of $f$, we have
\begin{align}\nonumber
f(Y)&=f(S\cup (Y\setminus S))+f(S\cap (Y\setminus S))\cr 
&\leq f(S)+f(Y\setminus S)\cr 
&\leq (|Y|-|S|+1)f(S).
\end{align}
Now we can write
\begin{align}\nonumber
f(\hat{\boldsymbol{y}})&=f(S)\ge \frac{1}{|Y|-|S|+1}f(Y)\cr 
&\ge \frac{1}{c+1}f(Y)=\frac{1}{c+1}f(\boldsymbol{y})\cr 
&=\frac{1}{c+1}f^M(\boldsymbol{y})\ge \frac{1}{c+1}f^M(\boldsymbol{y}^*)\cr 
& \ge  \frac{1}{c+1}(1-\frac{1}{e})f(OPT),
\end{align}
where $\boldsymbol{y}^*$ is the solution obtained from \eqref{eq:sub-y-relaxed} using the continuous greedy algorithm, and $OPT$ is the optimal integral solution to \eqref{eq:sub-y-relaxed}. Here, the third equality holds by integrality of $\boldsymbol{y}$, and the third inequality holds by the property of the pipage rounding that rounds a fractional solution $\boldsymbol{y}^*$ to an integral one $\boldsymbol{y}$ without decreasing the multilinear objective value $f^M$. Finally, the last inequality uses Lemma \ref{lemm:submodular-multilinear}.\hfill{$\blacksquare$}

\begin{remark}
In fact, one can bound the number of skipped elements $c$ in Algorithm \ref{alg:greedy}. As a naive upper bound, we note that selecting each new element into $S$ can eliminate the possibility of choosing at most $2n$ other elements into $S$. Since $Y$ has at most $n^2m$ elements, this gives an upper bound of $c\leq \frac{nm}{2}$. However, in practice, we observed that the value of $c$ in Algorithm \ref{alg:greedy} is much smaller than this naive upper bound. Although this bound depends polynomially on $n$ and $m$, since we are working with general positive concave externalities (submodular) functions, we believe that in the worst-case scenario, any approximation algorithm should have a polynomial or logarithmic dependence on these parameters. Nevertheless, improving this dependence on the parameters $n$ and $m$ is an interesting future research direction.    
\end{remark}

\section{Conclusions}\label{sec:conclusion}
We studied the maximum social welfare problem with multiple items subject to network externalities. We first showed that the problem could be cast as a multi-agent submodular or supermodular optimization. We then used convex programming and various randomized rounding techniques to devise improved approximation algorithms for that problem. In particular, we provided a unifying method to devise approximation algorithms for the multi-item allocation problem using the rich literature from submodular optimization. Our principled approach not only recovers or improves some of the existing algorithms that have been derived in the past in an ad hoc fashion, but it also has the potential to be used for devising efficient algorithms with additional complicating constraints.

This work opens several avenues for future research. It is interesting to extend our results by incorporating extra constraints into the MSW problem. For instance, in cyber-physical network security, resources tend to be limited and only a constrained subset of agents could have access to security resources. Moreover, it would be interesting to see if the approximation factors developed in this work can be improved or whether matching hardness lower bounds can be established. Finally, one can study a dynamic version of the MSW where the influence weights or the externality functions may change over time.

\bibliographystyle{IEEEtran}
\bibliography{thesisrefs}
\end{document}